\documentclass[prx,twocolumn,english,superscriptaddress,floatfix]{revtex4-2}
\pdfoutput=1

\usepackage[T1]{fontenc}
\usepackage[utf8]{inputenc}
\usepackage{caption}
\usepackage{amsmath, amsthm, amssymb,amscd, mathrsfs, amsfonts, mathtools,tikz-cd}
\usepackage{quantikz}
\usepackage{float} 
\usepackage{dsfont}
\usepackage{relsize}
\usepackage{lipsum}
\usepackage{pbox}
\usepackage[all]{xy}
\usepackage{todonotes}
\usepackage{xcolor}
\usepackage[bookmarks=true,
bookmarksnumbered=true, breaklinks=true,
pdfstartview=FitH, hyperfigures=false,
plainpages=false, naturalnames=true,
colorlinks=true,pagebackref=true,
pdfpagelabels]{hyperref}

\usepackage{tikz}
\usepackage{pgfplots}

\pgfplotsset{compat=1.10}
\usepgfplotslibrary{fillbetween}
\usetikzlibrary{patterns}
\usetikzlibrary{matrix}

\usepackage[]{amsrefs}

\renewcommand{\PrintDOI}[1]{\href{http://dx.doi.org/\detokenize{#1}}{doi: \detokenize{#1}}%
	\IfEmptyBibField{pages}{, (to appear in print)}{}}

\def\commutatif{\ar@{}[rd]|{\circlearrowleft}}

\newtheorem{thm}{Theorem}[section]
\newtheorem*{sch}{Schur's Lemma}
\newtheorem*{MT}{Main Theorem}

\newtheorem{prop}[thm]{Proposition}

\newtheorem{cor}[thm]{Corollary}

\theoremstyle{definition}
\newtheorem{defn}[thm]{Definition}

\theoremstyle{remark}
\newtheorem{rmk}[thm]{Remark}

\newtheorem{ex}[thm]{Example}

\allowdisplaybreaks

\newcommand{\C}{\mathbb{C}}

\usepackage[draft, commentmarkup=todo,
  todonotes={textsize=tiny, textwidth=0.83in}]{changes}

\definechangesauthor[name={Ilya Safro}, color=yellow]{IS}

\definechangesauthor[name={Boris}, color=blue]{BT}

\newcommand\minOne{\operatorname{-}1}
\newcommand\minTwo{\operatorname{-}2}
\newcommand\redBull{\color{red}{\bullet}}

\usepackage{bbm}

\hypersetup{
	colorlinks = true,
	urlcolor = blue,
	linkcolor = blue,
	citecolor = red,
	pdfpagemode = UseNone
}

\usepackage{youngtab}
\usepackage{ytableau}

\begin{document}

\author{Boris Tsvelikhovskiy} 
\affiliation{Department of Mathematics, University of California, Riverside, CA, USA} 

\author{Ilya Safro} 
\affiliation{Department of Computer and Information Sciences, University of Delaware, Newark, DE, USA} 

\author{Yuri Alexeev} 
\affiliation{Computational Science Division, Argonne National Laboratory, Argonne, IL, USA} 

\title{Symmetries and Dimension Reduction in Quantum Approximate Optimization Algorithm}

\begin{abstract}
In this paper, the Quantum Approximate Optimization Algorithm (QAOA) is analyzed by leveraging symmetries inherent in problem Hamiltonians. 
We focus on the generalized formulation of optimization problems defined on the sets of $n$-element $d$-ary strings. 
Our main contribution encompasses dimension reductions for the originally proposed QAOA. These reductions retain the same problem Hamiltonian as the original QAOA but differ in terms of their mixer Hamiltonian, and initial state. The vast QAOA space has a daunting dimension of exponential scaling in $n$, where certain reduced QAOA spaces exhibit dimensions governed by polynomial functions. This phenomenon is illustrated in this paper, by providing partitions corresponding to polynomial dimensions in the corresponding subspaces. As a result, each reduced QAOA partition encapsulates unique classical solutions absent in others, allowing us to establish a lower bound on the number of solutions to the initial optimization problem. Our novel approach opens promising practical advantages in accelerating the algorithm. Restricting the algorithm to Hilbert spaces of smaller dimension may lead to significant acceleration of  both quantum  and classical simulation of circuits and serve  as a  tool to cope with barren plateaus problem.
\end{abstract}

\maketitle

\hspace{-0.17in}\textbf{Keywords:} quantum approximate optimization algorithm, mixer Hamiltonians, representation theory
 
\section{Introduction}

Over the past few years, the Quantum Approximate Optimization Algorithm (QAOA) has emerged as a promising approach for tackling combinatorial optimization challenges on quantum computers. Moreover, it was recently established that the algorithm can be used to perform universal quantum computation (see \cite{Ll,MBZ}). However, the efficacy of QAOA crucially hinges on the careful design of its parameterized ansatz, which can be a non-trivial task, particularly for complex problem instances.  In a recent work by Hastings \cite{Has}, it is shown that a certain classical local algorithm outperforms the one-step  QAOA for the MAX-CUT problem on triangle-free graphs for all choices of degree. The author also claims that, for certain problems, increasing the number of steps, $p$, by a bounded amount may not yield substantial enhancements in performance. Other numerical studies on small problem instances also suggest that surpassing small values of $p$ may be necessary to achieve potential quantum advantage (see \cite{SA, WHJR, ZWCPL}).  On the other hand, as the depth increases, complications may arise, necessitating new approaches to managing resources and optimizing performance.

Advances in QAOA have shown promising strides in enhancing its performance on various combinatorial optimization tasks: machine learning approaches \cite{KSCAB,VBMSBJNM}, QAOA parameter transferability \cite{simtransfer}, adjusting initial distribution \cite{kulshrestha2022beinit}, multistart  \cite{shaydulin2019multistart} and warm-start \cite{egger2021warm} methods 
are among the techniques to improve its performance. To tackle large-scale optimization problems, QAOA was used as a component of multi- and single-level decomposition algorithms \cite{ushijima2021multilevel,shaydulin2019hybrid}. A more recent technique, pioneered in \cite{GHPGA, PGAG}, involves leveraging the Yang-Baxter equation for quantum circuit compression. These combined efforts mark significant progress toward harnessing the potential of QAOA in solving complex optimization problems. All these avenues are extremely important as QAOA is massively proposed as a universal strategy solver for applications from finance \cite{herman2023quantum} to biology \cite{boulebnane2023peptide}.

Another promising strategy for addressing the challenges and constraints mentioned earlier involves utilizing the symmetries intrinsic to QAOA, which encompass the symmetries of the target classical optimization problem and extend beyond them. Symmetries have played a foundational role not only in classical physics and mathematics but also in the realm of quantum computing. In quantum information theory, symmetries have been instrumental in the development of robust error correction codes, which are critical for fault-tolerant quantum computation \cite{Got}. Furthermore, in the study of quantum many-body systems, symmetries are central to the classification and characterization of phases of matter, providing valuable insights into the diverse array of quantum phenomena \cite{CGLW}. Their significance extends to quantum cryptography, where symmetries underlie the security of protocols such as Quantum Key Distribution (QKD) \cite{SBCDLP}. In this context, our investigation into symmetries within the QAOA is a natural progression, capitalizing on their established importance in quantum computing methodologies. 

\textit{Our contribution.}
We present a general approach to harness the principles of representation theory for groups of symmetries of the problem Hamiltonian into  facilitating various \textit{reductions} of  the original QAOA. Each reduced QAOA operates within a subspace of the initial ambient Hilbert space, featuring its own mixer Hamiltonian. While the ambient Hilbert space $W$ has a dimension of $d^n$,  exponential in $n$, some subspaces exhibit dimensions determined by polynomial functions in $n$.  \textit{The explicit construction is elaborated in Theorem \ref{MainObs}.}  Our main contributions are summarized as follows.

\begin{enumerate}
    \item We develop the framework for selecting the mixer Hamiltonian based on the symmetries inherent in the optimization problem under consideration.  This methodology, to the best of our knowledge, has not been explored in the literature.

    \item In instances where the group of symmetries includes the symmetric or unitary group, we construct reduced QAOAs for each subspace $W_\lambda$ appearing in decomposition of $W$ based on the group action.
    
    \item Each obtained reduction, denoted as $QAOA_\lambda$ (here $\lambda\vdash n$ is a partition), maintains the same problem Hamiltonian $H_P$ as the original QAOA. However, they differ in their choice of mixer Hamiltonian $H_{M,\lambda}$ and initial state $\xi_\lambda$, residing within the subspace $W_{\lambda}\subset W$, the explicit expressions for both are provided and the respective properties are rigorously verified.
    
    \item We show that, notably, while the ambient Hilbert space $W$ boasts a dimension of $d^n$, exponential in $n$, this is not the case for certain $W_\lambda$'s. For instance, Example \ref{PolDim} presents a set of partitions where the corresponding subspaces have dimensions determined by polynomial functions in $n$.
    
    \item Furthermore, we prove that each reduced QAOA encompasses at least one classical solution not found in the others. This observation enables us to establish a lower bound on the number of solutions to the initial optimization problem (see Corollary \ref{MinSolCor}).
\end{enumerate}

The exposition in the paper is organized as follows. In Section $2$, we provide a comprehensive overview of the main results. Sections $3$ to $7$ offer concise summaries of both the Quantum Approximate Optimization Algorithm and key concepts in representation theory pertinent to this paper's objectives. The sections offer references for a more comprehensive understanding of these concepts. 
Section $6$ contains an important, and to the best of our knowledge, novel result. It offers an elucidation of the groups of symmetries associated with problem Hamiltonians stemming from optimization problems defined on  $d$-ary strings of length $n$ with a quadratic objective function (see Theorem \ref{HamThm}). In particular, this result immediately implies (Corollary \ref{HamThmCor}) that unless the problem Hamiltonian is constant, its group of symmetries will never include the subgroup of unitary matrices $SU_d(\mathbb{C})$. This section concludes with the general formulation of our main result (Theorem \ref{MainObs}).
In Section $8$, we examine specific instances of Theorem \ref{MainObs}, and explore potential applications in Section $9$. Detailed proofs the theorems and corollaries are deferred to the Appendix.

\section{Overview of the approach}
Let $\mathbb{D}^n:=\{0,1,\hdots,d-1\}^n$ be the set of $n$-element $d$-ary strings and let $\mathcal{S}$ denote the group of permutations acting on these $d^n$ elements. A large class of optimization problems consists of finding the elements in $\mathbb{D}^n$. A significant class of optimization problems entails finding  elements in $\mathbb{D}^n$ on which a given function 
\begin{equation}\label{eq:optfunc}
F: \mathbb{D}^n \rightarrow \mathbb{R}
\end{equation}
(i.e., the optimization objective with possible constraints) attains either minimum or maximum values.

Consider the vector space $W$  of dimension $d^n$ with basis $\{v_x\}_{x\in \mathbb{D}^n},$ indexed by elements of $\mathbb{D}^n$. The Hamiltonian $H_F$ is said to \textit{represent} a function $F:\mathcal{S}\rightarrow \mathbb{R}$ if it satisfies $H_F(v_{x})=F(x)v_x$ for any $x\in\mathcal{S}$. The quantum 'analogue' of the original problem of optimizing  $F$ is given by the following classical-quantum 'dictionary':

\begin{itemize}
	\item set $\mathbb{D}^n\rightsquigarrow$ vector space $W$;
	\item objective function $F\rightsquigarrow$ linear operator $H_F$ acting on $W$;
	\item minima of $F$ on $\mathbb{D}^n\rightsquigarrow$ lowest energy states  of $H_F$ in $W$.
\end{itemize}

One of the most frequently used algorithms for solving the aforementioned quantum version of the original problem (\ref{eq:optfunc}) is the Quantum Approximate Optimization Algorithm (QAOA), introduced in \cite{QAOA}. In the framework of QAOA, the Hamiltonian $H_F$ is usually referred to as the \textit{problem Hamiltonian} and is denoted by $H_P$ (as in Farhi's et al. paper \cite{QAOA}). We will adhere to this notation.

The essential component of QAOA is a mixer Hamiltonian $H_M$, which possesses a unique lowest energy state $|\xi\rangle \in W$ and also satisfies the requirements of the Perron-Frobenius theorem (see Theorem \ref{PF}). The main idea behind the QAOA algorithm is to perform a multistep transformation of the mixer Hamiltonian $H_M$ into the problem Hamiltonian, such that the image of the lowest-energy vector from the preceding step becomes the lowest-energy vector in the subsequent step.

The algorithm begins by preparing the state $|\xi\rangle$, which is the ground state for the mixer Hamiltonian $H_M$, and then proceeds with multiple alternating applications of (certain exponents of) the problem and mixer Hamiltonians. The number of iterations is commonly denoted by $p$ (also known as QAOA depth). We will use $\mathcal{G}$ to express the entire composition of operators
\begin{equation}\label{qaoa-chain}
\mathcal{G}:=e^{-i\beta_1 H_M}e^{-i\gamma_1 H_P}\hdots e^{-i\beta_p H_M}e^{-i\gamma_p H_P}.
\end{equation}

The final step of QAOA involves performing a measurement of the state 
obtained after applying $\mathcal{G}$  in the standard basis. For a more detailed description of the algorithm, we refer the reader to Section $2$ and the references therein.

While the Hamiltonian $H_P$, representing the objective function, is uniquely determined by the classical problem (see Remark \ref{uniqueHam}), there is some flexibility in the choice of the mixer Hamiltonian. The convergence of QAOA to a classical state representing an element on which $F$ attains a minimum value is guaranteed by the adiabatic theorem, as long as the mixer Hamiltonian satisfies the conditions of the Perron-Frobenius theorem (see Theorem \ref{PF}). The standard and most common choice of mixer Hamiltonian is the one comprised of Pauli $X$-gates $X_j$, namely, $H_M=\sum\limits_{0\leq j \leq \ell-1} X_j$, where $\ell$ is the number of qubits required for the (re)formulation of the original problem. While this choice has some advantages, it does not take into account any special attributes of a specific problem.

The choice of mixer Hamiltonian has been discussed in the literature. In \cite{HWORVB}, the authors introduced a quantum alternating operator ansatz to allow more general families of Hamiltonian operators. The mixers in that article are useful for optimization problems with hard  constraints that must always be satisfied (thus defining a feasible subspace of $W$) and soft constraints whose violation needs to be minimized.

In \cite{GPSK}, it was experimentally verified (via numerical simulations) that linear combinations of $X$- and $Y$-Pauli gates as mixers can outperform the standard low depth QAOA. More examples can be found in  \cite{BFL,GBOE,ZLCMBE} and subsequent references.

Let $G$ be the group of symmetries of the problem Hamiltonian. Recently, it has been observed that 'taking the $G$-action into account' allows for both considering reductions of QAOA to subspaces of $W$ and increasing fidelity (see \cite{SHHS}, \cite{SG}, and \cite{SW}). It is also natural to exploit this group of symmetries and choose a mixer Hamiltonian that \textit{behaves well} with respect to the $G$-action on $W$ (preserves decomposition \eqref{generalDecomp} below). We will focus on the design of such mixers. Now, let us elaborate on the properties we would like $H_M$ to have.

If a permutation $g \in \mathcal{S}$ is \textit{undetectable} by $F$, i.e., $F(g(x)) = F(x)$ for any $x \in \mathbb{D}^n$, then $g$ is called a \textit{symmetry} of $F$. Such symmetries form a subgroup $\widetilde{G} \subseteq \mathcal{S}$. The set $\mathbb{D}^n$ can be expressed as a disjoint union of $\widetilde{G}$-orbits: 
$$\mathbb{D}^n = \bigcup_{j=1}^{m} \mathcal{O}_j.$$ 
The action of $\widetilde{G}$ on the set $\mathbb{D}^n$ extends to an action on $W$. The actions of $\widetilde{G}$ and $H_P$ on $W$ commute:
$$H_P(g(w)) = g(H_P(w)) \quad \forall w \in W, \quad \forall g \in \widetilde{G}.$$
Therefore, $W$ is a \textit{representation} of $\widetilde{G}$ (see Section 3.1) and can be expressed as a direct sum of subspaces:
\begin{equation}
    W = \bigoplus_i W_i,
    \label{generalDecomp}
\end{equation}
where the index $i$ enumerates irreducible representations $V_i$ of $\widetilde{G}$ appearing in $W$ and each $W_i = V_i^{\oplus m_i}$. We suggest that $H_M$ should be chosen in such a way that it preserves the decomposition in \eqref{generalDecomp}. For instance, it suffices for $H_M$ to lie in the centralizer of $G$ to fulfill this requirement. In case the initial vector $|\xi\rangle$ is chosen to be inside one of the $W_i$'s, all the iterations of QAOA will keep it within $W_i$. In other words, QAOA can be \textit{restricted} to $W_i$, which in certain cases has significantly lower dimension than $n^d$, the dimension of the ambient space $W$ (see Figures \ref{OldPic} and \ref{NewPic} for a schematic illustration).

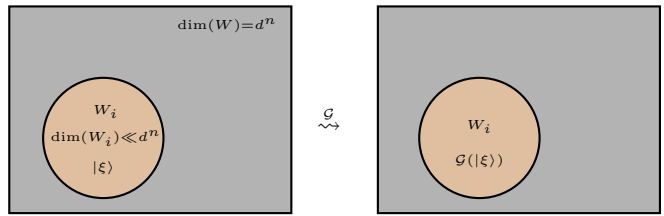
\begin{figure}[htbp!]
	\begin{center}	
		\begin{tikzpicture}[scale=0.5,thick]
			\draw[thick, fill=black!30] (-2.5,-2) -- (5,-2) -- (5,3.5) -- (-2.5,3.5) -- cycle;
			
			\draw[fill=brown!50] (0,0) circle(1.6cm);
			\node (A) at (3.3,3) {$\mathsmaller{\mathsmaller{\dim(W)=d^n}}$};
			
			\node (C) at (6,0.5) {$\overset{\mathsmaller{\mathsmaller{\mathcal{G}}}}{\rightsquigarrow}$};
			
			\node (B) at (0.07,0) {$\mathsmaller{\mathsmaller{\dim(W_i)\ll d^n}}$};
			
			\node (B) at (0.07,0.7) {$\mathsmaller{\mathsmaller{W_i}}$};
			
			\draw[thick, fill=black!30] (7.3,-2) -- (14.8,-2) -- (14.8,3.5) -- (7.3,3.5) -- cycle;
			\draw[fill=brown!50] (10,0) circle(1.6cm);
			\node at (0,-0.8) {$\mathsmaller{\mathsmaller{|\xi\rangle}}$};
			
			\node (B) at (10,0.3) {$\mathsmaller{\mathsmaller{W_i}}$};
			
			\node at (10,-0.6) {$\mathsmaller{\mathsmaller{\mathcal{G}(|\xi\rangle)}}$};
			
		\end{tikzpicture}
	\end{center}
 \caption{QAOA in the ambient space}
 \label{OldPic}
\end{figure}
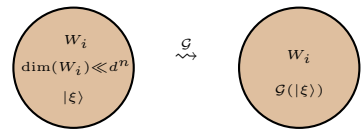
\begin{figure}[htbp!]
	\begin{center}	
		\begin{tikzpicture}[scale=0.5,thick]

			\draw[fill=brown!50] (2,0) circle(1.6cm);
			
			\node (C) at (5,0.5) {$\overset{\mathsmaller{\mathsmaller{\mathcal{G}}}}{\rightsquigarrow}$};
			
			\node (B) at (2.07,0) {$\mathsmaller{\mathsmaller{\dim(W_i)\ll d^n}}$};
			
			\node (B) at (2.07,0.7) {$\mathsmaller{\mathsmaller{W_i}}$};
			
			\draw[fill=brown!50] (8,0) circle(1.6cm);
			\node at (2,-0.8) {$\mathsmaller{\mathsmaller{|\xi\rangle}}$};
			
			\node (B) at (8,0.3) {$\mathsmaller{\mathsmaller{W_i}}$};
			
			\node at (8,-0.6) {$\mathsmaller{\mathsmaller{\mathcal{G}(|\xi\rangle)}}$};
			
		\end{tikzpicture}
	\end{center}
 \caption{QAOA restricted to a subspace of smaller dimension}
 \label{NewPic}
\end{figure}

Alternatively,  $G$ can be defined as the subgroup of elements whose action on $W$ commutes with that of $H_P$: $$\{g\in SU(W)~|~[g,H_M]=[g,H_P]=0\}.$$

It is important to point out that while $\widetilde{G}$ is a finite group, $G$ is usually infinite. On the classical level one can only detect $\widetilde{G}$, which is a refinement (subgroup) of $G$.  We give an explicit description of $G$ in case the problem Hamiltonian $H_P$ represents a quadratic objective function 
\begin{equation}
    F = a + \sum_{k=1}^n \beta_k \mathbf{x_k} + \sum_{1\leq i\leq j\leq n} \alpha_{ij}\mathbf{x_ix_j}
\label{eq:F_definition}
\end{equation}
with coefficients $a$, $\alpha_{ij}$, $\beta_k \in \mathbb{R}$ and variables $\mathbf{x_i}\in\{0,1,\ldots, d-1\}$. The corresponding result can be found in Theorem \ref{HamThm}. In case $d=2$, the (classical) optimization problem is known as quadratic unconstrained binary optimization (QUBO).

Notice that $W$ has a natural structure of an $n$-fold tensor product of $n$ copies of $V_{\mathbb{D}}$, a $d$-dimensional vector space representing possible states of a single qudit: $W=V_{\mathbb{D}}^{\otimes n}$. Furthermore, the unitary group $SU_d(\mathbb{C})$ and symmetric group $S_n$ naturally act on $W$:

\begin{itemize}
    \item $g(v_1\otimes\hdots\otimes v_n):=g(v_1)\otimes\hdots\otimes g(v_n)$ for $g\in SU_d(\mathbb{C})$ and
    \item $\sigma(v_1\otimes\hdots\otimes v_n):=v_{\sigma(1)}\otimes\hdots\otimes v_{\sigma(n)}$ for $\sigma\in S_n$.
\end{itemize}

The two actions commute with each other. It so happens that in case  $G$ contains one of the subgroups $S_n$ or $SU_d(\mathbb{C})$, the decomposition in Equation \eqref{generalDecomp} admits an illuminating description:
\begin{equation}
W=\underset{r(\lambda)\leq d}{\underset{\lambda\vdash n}{\bigoplus}} W_{\lambda},
\label{SWdecomp}
\end{equation}
where $\lambda$ runs over partitions of $n$ with at most $d$ constituents. This result is known under the name of Schur-Weyl duality and is the main topic of Section $5.2$.

Motivated by the decomposition in Equation \eqref{SWdecomp}, we mainly focus on the cases when $G$ contains one of the groups $S_n$ or $SU_d(\mathbb{C})$. 
Notice that the classical refinement of $G=SU_d(\mathbb{C})$ is its Weyl subgroup $\widetilde{G}=S_d\subset SU_d(\mathbb{C})$ (acting on $\mathbb{D}^n$). Considering the latter group of symmetries would be insufficient to infer the existence of decomposition presented in \eqref{SWdecomp}.

In this work, we introduce mixer Hamiltonians $H_{M,\lambda}$ corresponding to each $W_\lambda$ as defined in equation \eqref{SWdecomp}. These mixer Hamiltonians satisfy the conditions appearing in Theorem \ref{MainObs} and maintain the decomposition described in \eqref{SWdecomp}. The explicit descriptions of both the mixer Hamiltonian and the vector spanning its one-dimensional lowest-energy subspace can be found in Theorem \ref{MainThm}.

\section{QAOA review}
In this paper, by a classical optimization problem (COP) we will understand
a maximization or minimization of $F$ from Equation (\ref{eq:optfunc}). 
The Quantum Approximate Optimization Algorithm (QAOA), initially proposed in \cite{QAOA} is a hybrid quantum-classical algorithm that combines a parameterized quantum evolution with a classical parameter optimization. The outcome of such algorithm is an (approximate) solution to optimization problem as above.

\subsection{Generalities} We initiate our discussion with an overview of the widely adopted QAOA implementation. Let $W$ denote a $d^n$-dimensional vector space, possessing a basis $\{v_x\}_{x\in \mathbb{D}^n}$, where the indexing is defined over the elements of $\mathbb{D}^n$.

\begin{defn}
A Hamiltonian $H_F$ is said to \textbf{represent} a function $F:\mathcal{S}\rightarrow \mathbb{R}$ if it satisfies $H_F(v_{x})=F(x)v_x$ for any $x\in\mathcal{S}$.
\end{defn}
\begin{rmk}
The set of conditions $H_F(v_{x})=F(x)v_x$ for any $x\in\mathcal{S}$ and linearity of $H_F$ imply that the Hamiltonian representing $F$ is unique.
\label{uniqueHam}
\end{rmk}

The steps of the QAOA approach are as follows:

\begin{enumerate}
\item[\textbf{Step 1.}] Define a mixer Hamiltonian $H_M$.
\item[\textbf{Step 2.}] Construct the quantum circuits representing $e^{-i\beta H_M}$ and $e^{-i\gamma H_P}$ with $\beta \in [0,\pi)$ and $\gamma\in [0,2\pi)$. These are known as \textit{problem} (or cost) and \textit{mixer} operators, respectively.
\item[\textbf{Step 3.}] Choose a number $p\in \mathbb{Z}_{\geq 1}$ (the QAOA depth)  and build the circuit consisting of repeated application of the $\{e^{-i\beta_j H_M}e^{-i\gamma_j H_C}\}_{j=1,\hdots, p}$ evolution.
\item[\textbf{Step 4.}] Prepare an initial state $|\xi\rangle$ (which has minimal energy w.r.t. the mixer Hamiltonian).

\item[\textbf{Step 5.}] Optimize the array of parameters $[\beta_1,\gamma_1,\hdots,\beta_p,\gamma_p]$ using classical computer. Repeat steps $1$-$5$.

\item[\textbf{Step 6.}] After the circuit has been optimized, measurements of the output state in the standard basis reveal approximate solutions to the optimization problem.
\end{enumerate}

\subsection{Mixer Hamiltonians}
We would like to discuss mixer Hamiltonians in more detail. A mixer Hamiltonian can be treated as a starting operator, which is continuously transformed into the problem Hamiltonian. The  transformation  is produced by the family of operators $tH_P+(1-t)H_M$ for $0\leq t\leq 1$, i.e. a linear interpolation between $H_M$ and $H_P$. It is natural to see which linear operators can serve as mixer Hamiltonians for a given optimization problem. In other words, we would like to know the technical requirements on the mixer Hamiltonian, so that the corresponding QAOA converges to a classical solution as the number of iterations increases.  
\begin{defn}
An $n\times n$ matrix $M$ with real coefficients is called \textbf{irreducible} iff there are no proper $M$-invariant coordinate subspaces of $\mathbb{C}^n$: there is no $0\neq W\subset \mathbb{C}^n$ with $W$ a coordinate subspace and  $M(W)\subseteq W$.
\end{defn}

\begin{thm} (Perron-Frobenius). Let $M=(m_{ij})\in \mbox{Mat}_n(\mathbb{C})$ be an irreducible matrix with $m_{ij}\geq 0$. Then there is a positive real number $r$, such that $r$ is an eigenvalue of $M$ and any other eigenvalue $\lambda$ (possibly complex) has $\operatorname{Re}(\lambda)<r$.
\label{PF}
\end{thm}

The properties of mixer Hamiltonian sufficient for convergence of corresponding QAOA are exactly those appearing in Theorem \ref{PF}  (see \cite{BKZS}).

\section{Representation theory overview}
The main subject of representation theory is to find  the ways in which a given group $G$ may be embedded in a linear group $GL(V)$. The vector space $V$ can be of finite or infinite dimension, and it may be defined over a field of characteristic zero or a finite field. However, in the context relevant to this article all vector spaces will be finite dimensional and over the field of complex numbers, $\mathbb{C}$. The overview below is very brief, yet contains all results required for understanding the material in subsequent sections. A more detailed and rigorous exposition on the subject can be found in Chapter $4$ in \cite{FH}, Section $4$ in \cite{Et}. 


\begin{defn}
The following are basic definitions in the representation theory:
\begin{itemize}
\item A \textbf{representation}  of a group $G$ (or a $G$\textbf{-module}) on a finite-dimensional complex vector space $V$ is a homomorphism $\rho: G\rightarrow GL(V)$ of $G$ to the group of \textit{automorphisms} of $V$ (invertible linear operator on $V$).
    \item A map $\varphi$ between two $G$-representations ($G$\textbf{-module homomorphism}) $(\rho_1,V)$ and $(\rho_2,W)$ is a vector space map
$\varphi: V\rightarrow W$, such that the diagram

\[\begin{tikzcd}
	{V} && {W} \\
	\\
	{V} && {W}
	\arrow["{\rho_1(g)}"', from=1-1, to=3-1]
	\arrow["\varphi"', from=3-1, to=3-3]
	\arrow["{\rho_2(g)}", from=1-3, to=3-3]
	\arrow["\varphi", from=1-1, to=1-3]
\end{tikzcd}\]
commutes for every $g \in G$.
\item A \textbf{subrepresentation} of a representation $(\rho, V)$ is a vector subspace $W\subset V$, which is invariant under the action of $G$. 
\item A representation $V$ is called \textbf{irreducible} if there is no proper nonzero $G$-invariant subspace $W \subset V$.
\end{itemize}

\end{defn}

Let $V$ and $W$ be representations of a group $G$. Then the direct sum $V\oplus W$ and tensor product
$V\otimes W$ are $G-$representations as well: $g(v+w)=g(v)+g(w)$ and $g(v\otimes w)=g(v)\otimes g(w)$ for any $v\in V, w \in W$ and $g\in G$.

\begin{defn}
    Two elements $a,b \in G$ are \textbf{conjugate} if there is an element $g\in G$ such that $b=gag^{-1}$. This is an equivalence relation whose equivalence classes are called \textbf{conjugacy classes}.
\end{defn}

The following results can be found as Proposition $2.30$ and Proposition $1.8$ in \cite{FH}.

\begin{prop} Let $G$ be a finite group.
\begin{enumerate}
\item Irreducible representations are in bijection with conjugacy classes in $G$.
    \item  Every finite-dimensional representation $V$ of $G$ over a field of characteristic zero  (in particular, over $\mathbb{C}$) is completely reducible: $V\simeq\bigoplus V_i^{\oplus m_i}$, where each $V_i$ is irreducible.
\end{enumerate}
 \label{basicFacts}     
\end{prop}

We will need one more important fact (see Lemma $1.7$ in \cite{FH}).

\begin{sch}
Let $V$ and $W$ be irreducible representations of $G$ and $\varphi: V\rightarrow W$ a $G$-module homomorphism, then
\begin{enumerate}
    \item $\varphi$ is an isomorphism or zero;
    \item in the former case $\varphi=\lambda\cdot I$ for some $\lambda\in\mathbb{C}$.
\end{enumerate}
\label{Schur}
\end{sch}

\section{QAOA dimensionality reduction}
The symmetric group $\mathcal{S}$ acts on  the set of $n$ qudits $\mathbb{D}^n:=\{0,1,\hdots,d-1\}^{n}$ by permutations. This action can be extended to a linear action on the vector space  $W:=\underset{s\in \mathbb{D}^n}{\bigoplus}\mathbb{C}\langle s\rangle$. Said differently, there is a homomorphism $\varphi: \mathcal{S} \rightarrow GL(W)$.  

 Let 
 \[
 G:=\{g\in SU(W)~|~[g,H_M]=[g,H_P]=0\}
 \]
 be the subgroup of elements whose action on $V$ commutes with that of both Hamiltonians $H_M$ and $H_P$. Then 
\begin{equation}
(e^{-i\beta H_M}e^{-i\beta H_P})^p(\varphi(g)v)=\varphi(g)((e^{-i\beta H_M}e^{-i\beta H_P})^p(v))
\label{commutation}
\end{equation} 
for any $v\in W$. Let
\[
\widetilde{G}:=\{g\in \mathcal{S}~|~[g,H_M]=[g,H_P]=0\}\subseteq \mathcal{S}
\]
be the subgroup acting on the set $\mathbb{D}^n$ and commuting with $H_M$ and $H_P$. In other words, $\widetilde{G}$ is the 'refinement' of $G$  after transition from the vector space $V_{W}$ back to the underlying set $\mathbb{D}^n$ indexing its basis. It is natural to refer to the groups $G$ and $\widetilde{G}$ as the \textit{quantum} and \textit{classical symmetries} of the problem Hamiltonian $H_P$.  We first turn our attention to the action of the group of classical symmetries. The set $\mathbb{D}^n$  can be written as a disjoint union of $\widetilde{G}$-orbits: $\mathbb{D}^n=\overset{m}{\underset{j=1}{\bigsqcup}} \mathcal{O}_j$. Set $\xi_j:=\frac{1}{\sqrt{|\mathcal{O}_j|}}\sum\limits_{s\in \mathcal{O}_j}|s\rangle$ and observe that $\xi_j$ is $\widetilde{G}$-invariant (each element $g\in\widetilde{G}$ permutes the summands of $\sum\limits_{s\in \mathcal{O}_j}|s\rangle$). Moreover, the subspace of $\widetilde{G}$-invariants, 
\[
W^{\widetilde{G}}:=\{w\in W~|~g\cdot w=w ~\forall g\in \widetilde{G}\},
\]
is spanned by the vectors $\xi_1,\hdots,\xi_m$. As the subspace $(W)^{\widetilde{G}}$ is preserved by $H_M, H_P$ and, hence, $(e^{-i\beta H_M}e^{-i\beta H_P})^p$ for any $p\in Z_{> 0}$, we get that each vector 
\[
(e^{-i\beta H_M}e^{-i\beta H_P})^p(|\xi\rangle)=\sum\limits_{j=1}^{m}\alpha_j \xi_j
\]
is inside $ W^{\widetilde{G}}$ as well. Notice that the vector $\xi=\frac{1}{|\mathbb{D}^n|}\sum\limits_{s\in \mathbb{D}^n} |s\rangle$, which can be expressed as 
\[
\xi=\frac{1}{\sqrt{|\mathbb{D}^n|}}\sum\limits_{j=1}^{m} (|\mathcal{O}_j|\cdot|\xi_j\rangle),
\]
is invariant under the action of  $\mathcal{S}$.

\begin{rmk}
One can also describe $\widetilde{G}$ simply as the intersection of the groups $G$ and $\mathcal{S}$.
\end{rmk}


\begin{ex}
If the set $\mathbb{D}^n$ forms a single $\widetilde{G}$-orbit, all the eigenvalues of $H_P$ are equal. Therefore, the problem Hamiltonian is constant, and so is the objective function. For instance, this happens when $\widetilde{G}=\mathcal{S}$ or $\widetilde{G}$ contains a long cycle in $\mathcal{S}$.
\end{ex}

The homomorphism $\varphi: \widetilde{G}\rightarrow GL(W)$ makes $W$ into a $\widetilde{G}$-representation. As $\widetilde{G}$ is a finite group, every representation of $\widetilde{G}$ over a field of characteristic zero  (in particular, over $\mathbb{C}$) is completely reducible (see Proposition \ref{basicFacts}). The latter implies that $W\simeq\bigoplus V_i^{\oplus m_i}$, where each $V_i$ is irreducible. As the vector $\xi$ is $\widetilde{G}$-invariant, the one-dimensional space $\mathbb{C}\langle\xi\rangle$ is a trivial representation of $\widetilde{G}$. 

The key observation is that since the actions of $H_M$ and $H_P$ on $W$ commute with the action of $\widetilde{G}$, we get that $H_M$ and $H_P$ are not just linear maps, but  homomorphisms of $\widetilde{G}$-modules.

 It follows from Lemma \ref{Schur} that $H_M$ and $H_P$ can map irreducible representations only to irreducible representations of the same type. In particular, the image of $\xi$ under $H_M$ or $H_P$ and, therefore, $(e^{-i\beta H_M}e^{-i\beta H_p})^p$, must be inside $V_{0}^{\oplus m_{0}}\subseteq W$, where $W_0$ is the trivial representation of $\widetilde{G}$. More generally, we make the following observation.
\begin{rmk}
   Let $V_j$ be an irreducible representation of $\widetilde{G}$ that appears in decomposition \begin{equation}
    W = \bigoplus_j W_j,
    \label{generalDecomp2}
\end{equation} where $W_j:=V_j^{\oplus m_j}$. Then the image of $\xi$ under $(e^{-i\beta H_M}e^{-i\beta H_P})^k$, must be inside $W_j\subseteq W$ for any $k>0$.
\end{rmk}

The following results (reformulations thereof) appeared in \cite{SHHS}. 

\begin{prop}The following statements hold.
\begin{enumerate}
    \item The QAOA output probabilities are the same across all elements in the same $\widetilde{G}$-orbit.
    \item The number of $\widetilde{G}$-orbits on  $\mathbb{D}^n$ is 
\[
m=\dim(W^{\widetilde{G}})=\sum\limits_{x\in \mathbb{D}^n}\frac{1}{\widetilde{G}\cdot x}.
\]
\end{enumerate}
\label{Amplitudes}
\end{prop}

The preceding discussion establishes that the classical QAOA, prior to its final measurement, unfolds within the Hilbert subspace $W_0\subseteq W$. This is  a special case within a broader, more general result.

\begin{MT}
Consider a Constrained Optimization Problem (COP) with problem Hamiltonian $H_P$. Let $K$ be a group of symmetries of $H_P$ (not necessarily the group of all symmetries), and let $W = \underset{i}{\bigoplus} W_i$ be the decomposition of $W$ into irreducible representations of $K$ with multiplicities. If there exists a mixer Hamiltonian $H_{M,i}$ that meets the following criteria: 
\begin{itemize}
\item $H_{M,i}$ satisfies the assumptions of Theorem \ref{PF}
\item the one-dimensional lowest energy eigenspace of $H_{M,i}$ is inside  $W_i$ and
\item $H_{M,i}$ preserves the direct sum decomposition of $W$ in \eqref{generalDecomp2},
\end{itemize}
then one can establish a reduced Quantum Approximate Optimization Algorithm (QAOA) with the same problem Hamiltonian $H_P$, mixer Hamiltonian $H_{M,i}$ and initial state $\xi_i$. Thus defined QAOA has an ambient Hilbert space $W_i$.
\label{MainObs}
\end{MT}

\begin{proof}The initial two assumptions align precisely with those in Theorem \ref{PF}, affirming the suitability of the operator $H_{M,i}$ as a mixer Hamiltonian. The final condition guarantees that in the restricted QAOA, employing the problem Hamiltonian $H_P$, mixer $H_{M,i}$ and initial state $\xi_i$, all subsequent transformations of the initial vector are constrained within the subspace $W_i\subseteq W$ prior to the measurement of the resulting state. It is worth noting that the preservation of $W_i$ under the action of the (exponents of)  problem Hamiltonian is a direct consequence of Lemma \ref{Schur}. 
\end{proof}

\section{Group of symmetries for problem Hamiltonians}

In this section, we explore the groups of quantum symmetries. The Quadratic Unconstrained Optimization (QUO) is a subclass of COPs focused on finding maxima or minima of quadratic functions on the set of $d$-ary strings of length $n$. In other words, the problem is to find the element(s) $x^* \in \mathbb{D}^n$ which minimizes (or maximizes) the objective function of the form 
\begin{equation}
    F = a + \sum\limits_{k=1}^n \beta_k x_k + \sum\limits_{1 \leq i \leq j \leq n} \alpha_{ij}x_ix_j
\label{quadrobjfn}
\end{equation}
with $a, \alpha_{ij}, \beta_k \in \mathbb{R}$.

The goal of this section is to determine the group of quantum symmetries for problem Hamiltonians representing such objective functions. We recommend \cite{Had} for material on QUBOs (Quadratic Unconstrained Binary Optimization, i.e., a particular case of QUO with $\mathbb{D} = \{0,1\}$), corresponding problem Hamiltonians, and reformulations of the former as QAOA problems.

The problem Hamiltonian representing a function $F$ as in \eqref{quadrobjfn}, arising from a QUO problem, takes the form:
\[
H_P = a \cdot I + \sum\limits_{k=1}^n \beta_k \cdot \frac{I-Z_k}{2} + \sum\limits_{1 \leq i \leq j \leq n} \alpha_{ij} \cdot \frac{(I-Z_i)(I-Z_j)}{4},
\]
which can be simplified into:
\[\begin{aligned}
    & H_P = \left(a + \frac{1}{2} \sum\limits_{k=1}^n \beta_k + \frac{1}{4} \sum\limits_{1 \leq i \leq j \leq n}\alpha_{ij}\right) \cdot I +\\
    &+\sum\limits_{k=1}^n \left(\frac{1}{2}\beta_k + \frac{1}{4}\sum\limits_{i\neq k}\alpha_{i,k}\right) Z_k + \sum\limits_{1 \leq i \leq j \leq n}\alpha_{ij}Z_iZ_j,
\end{aligned} \]
and, after disregarding the scalar term, becomes:
\[
H_P=\sum_{k=1}^n \widetilde{\beta}_kZ_k + \sum_{i,j} \widetilde{\alpha}_{i,j}Z_iZ_j.
\]

\begin{thm}
The Hamiltonian  $H_P=\sum\limits_{k}\widetilde{\beta}_kZ_k+\sum\limits_{i,j}\widetilde{\alpha}_{i,j}Z_iZ_j$ has group of symmetries \[
G=\left(\begin{array}{ c|c|c|c }
    M_1 & 0 & \hdots & 0 \\
    \hline
    0 & M_2 & \ddots & 0 \\
    \hline
    \vdots & \ddots & \ddots  & \vdots\\
    \hline
    0 & 0 & \hdots & M_t
  \end{array}\right)\subseteq SU_N(\mathbb{C}),
  \]
  where the blocks correspond to distinct eigenvalues of $H_P$. Specifically, $t$ represents the count of unique eigenvalues, and the size of the matrix block $M_i$ matches the multiplicity of the eigenvalue  $\lambda_i$. 
  \label{HamThm}
\end{thm}

\begin{cor}
\begin{enumerate}
    \item The group of symmetries $G$ always contains the subgroup of $SU_N(\mathbb{C})$, which consists of matrices that are diagonal in the standard basis. In particular, $G$ is always infinite.
    \item The matrix for action of a generic $g\in SU_d(\mathbb{C})$ on $W$ in the standard basis does not contain zero entries. Therefore, $SU_d(\mathbb{C})\not\subseteq G$, unless $H_P$ is constant.
\end{enumerate}
\label{HamThmCor}
\end{cor}

\section{Representation theory of $S_n, SU_d(\mathbb{C})$ and Schur-Weyl duality}

In this section, we provide a brief overview of the key results concerning representations of the groups $S_n$ and $SU_d(\mathbb{C})$, as well as the decomposition of the tensor product $(\mathbb{C}^d)^{\otimes n}$ under the natural actions of these groups. This concept is known as 'Schur-Weyl duality'. For a more comprehensive and detailed treatment, readers are encouraged to refer to Chapters $4$, $15$, and $16$ in \cite{FH}, and also Sections $4.12$, $4.18$, and $4.19$ in \cite{Et}, along with the references provided therein.

The material presented in this section will be applied to derive results that are pertinent to QAOA in the following section. Further applications of representation theory in the context of QAOA (including simulations) and quantum machine learning can be found in \cite{Z1LLSK1}, \cite{Z1LLSK2} and \cite{K}.

\subsection{Basic facts on representations of $S_n$ and $SU_d(\mathbb{C})$}
\begin{defn}
A \textbf{partition} of a positive integer $n$  is a presentation of $n$ in the form  $n=\lambda_1+\lambda_2+\hdots+\lambda_s$, where $\lambda_i\in\mathbb{Z}_{>0}$, and $\lambda_i\geq \lambda_{i+1}$. We will use the notation $\lambda \vdash n$ for partition $\lambda=(\lambda_1,\hdots,\lambda_s)$.

To any $\lambda \vdash n$ we will attach a \textbf{Young diagram} $Y_\lambda$, which is the union of grid squares (boxes) $\{(i,-t)~|~1\leq i\leq s, 1\leq t\leq \lambda_i\}$ in the Cartesian plane. Clearly, $Y_\lambda$ is a collection of $n$ unit squares. 

A \textbf{Young tableau} $T_\lambda$ corresponding to diagram $Y_\lambda$ is obtained by placing numbers $1,\ldots,n$ in the squares of $Y_\lambda$.

A tableau is called \textbf{standard} (SYT) if the entries in each row and each column are strictly increasing and \textbf{semistandard} (SSYT) if the entries weakly increase along each row and strictly increase down each column.
\end{defn}

\begin{ex}
A standard and a semistandard tableaux of shape $\lambda=(3,2)$ appear on Figure \ref{Tableaux} below.
\begin{figure}[htbp!]
    \centering
    $\young(124,35)\hspace{0.3in} \young(113,25)$
    \caption{Tableaux of shape $\lambda=(3,2)$}
    \label{Tableaux}
\end{figure}
\end{ex}

\begin{thm}
Nonisomorphic irreducible representations of $S_n$ are in natural bijection with (enumerated by) partitions $\lambda\vdash n$. Irreducible $S_n$-representation corresponding to partition $\lambda$ will be denoted by $\mathbb{S}_\lambda$.
\end{thm}

It so happens that a natural basis of an irreducible representation $\mathbb{S}_\lambda$ is given by the elements of $SYT(\lambda)$ (see \cite{OV} and \cite{VO}). Next we present a formula for the dimension of such representation. 

\begin{defn}
For a box with coordinates $(i,j)$, in Young diagram $Y_\lambda$, let \textbf{the hook length at $(i,j)$} be the number of boxes $(a,b)\in Y_\lambda$ with $a=i$ and $b\leq j$ or $b=j$ and $a\geq i$. This number will be denoted by $h_\lambda(i,j)$.
\end{defn}

\begin{rmk}
The dimension of representation $\mathbb{S}_\lambda$ corresponding to partition $\lambda$ is given by 
\begin{equation}
\dim(\mathbb{S}_\lambda)=\frac{n!}{\prod\limits_{(i,j)\in Y_\lambda}h_\lambda(i,j)}
\label{HLformula}
\end{equation}
This equality is usually referred to as the hook length formula.
\label{dimRemark}
\end{rmk}

\begin{ex}
Consider the partition $\lambda=(4,3,2)$ of $9$ and the box $(1,-1)$ in the diagram $Y_\lambda$. Then the hook length $h_\lambda(1,0)$ is equal to $6$ (see Figure \ref{Hook}). The dimension of the corresponding representation of $S_{9}$ is equal to $\frac{9!}{6\cdot5\cdot3\cdot1\cdot4\cdot3\cdot1\cdot2\cdot1}=168$.

\begin{figure}[htbp!]
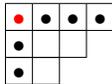

    \centering
    $\young(\redBull\bullet\bullet\bullet,\bullet~~,\bullet~)$
    \caption{Hook centered at box $(1,-1)$, depicted in red, inside diagram of shape $\lambda=(4,3,2)$}
    \label{Hook}
\end{figure}
    
\end{ex}

\begin{defn}
 The \textbf{Young-Jucys-Murphy elements}, denoted as YJM-elements, are defined by $\mathcal{J}_k=(1,k)+(2,k)+\hdots+(k-1,k)\in\mathbb{C}[S_n]$, where $2\leq k\leq n$. Additionally, $\mathcal{J}_1=0$.
\end{defn}

\begin{rmk}
The YJM elements generate a commutative subalgebra of the group algebra $\C[S_n]$.
\end{rmk}

\begin{defn}
     The \textbf{content} of a box with coordinates $(i,j)$ in a Young diagram $Y_\lambda$ is given by $c_\lambda(i,j):=i+j$.
\end{defn}

\begin{rmk}
Let $v_T\in \mathbb{S}_\lambda$ be the basis vector corresponding to $T\in SYT(\lambda)$. Then $v_T$ is an eigenvector for $\mathcal{J}_k$ with eigenvalue $c_\lambda(i_k,j_k)$, where $(i_k,j_k)\in T$ is the box with number $k$.
\label{YJMevals}
\end{rmk}

\begin{ex}
The content of each box in tableaux $T$ 
    $$\young(1246,35,7)$$ is given by $$\young(0123,\minOne0,\minTwo\minOne).$$ For instance, as number $3$ is assigned to the box with coordinates $(1,-2)$ and the contents of it is $-1$, we conclude that $\mathcal{J}_3(v_T)=-v_T$.
\end{ex}

\begin{thm}
Nonisomorphic irreducible representations of $SU_d(\mathbb{C})$ are in one-to-one correspondence with partitions $\lambda\vdash n$, where the corresponding Young diagram contains at most than $d$ rows.
\end{thm}

Irreducible $SU_d(\mathbb{C})$-representation corresponding to partition $\lambda$ will be denoted by $V_\lambda$.

\begin{rmk}
A natural basis of an irreducible representation $V_\lambda$ is indexed by the elements of $SSYT(\lambda)$ filled with numbers $\{1,\hdots,d\}$ (see \cite{GC}). The dimension of $V_\lambda$ can be computed as \begin{equation}
\dim(V_\lambda)=\prod\limits_{1\leq i<j\leq d}\frac{\lambda_i-\lambda_j+(j-i)}{j-i}
\label{DimFormula}
\end{equation}
\label{dimRmk2}
\end{rmk}
The latter result can be found as Theorem $6.3(1)$ in \cite{FH}.

\subsection{Schur-Weyl duality}
Let $V$ be a $d$-dimensional vector space and consider two actions on a tensor power $W=V^{\otimes n}$:
\begin{enumerate}
    \item The diagonal action of the unitary group $GL(V)$;
    \item The action of symmetric group $S_n$, achieved by permuting the factors.
\end{enumerate}

The actions of the two groups are given by 
\begin{itemize}
    \item $g(v_1\otimes\hdots\otimes v_n):=g(v_1)\otimes\hdots\otimes g(v_n)$ for $g\in GL(V)$ and
    \item $\sigma(v_1\otimes\hdots\otimes v_n):=v_{\sigma(1)}\otimes\hdots\otimes v_{\sigma(n)}$ for $\sigma\in S_n$
\end{itemize}
and commute with each other. 

\begin{defn}
The \textbf{centralizer} of a subalgebra $\mathcal{A}$ of an algebra $\mathcal{B}$ is the subalgebra of  $\mathcal{B}$ which consist of elements that commute with every element of $\mathcal{A}$:

 $$Z_\mathcal{B}(\mathcal{A}):=\{b\in \mathcal{B}~|~ab=ba \mbox{ for all } a \in \mathcal{A}\}.$$ 
\end{defn}

\begin{thm}
Let $A, B \subset \mathrm{End}(W)$ be the subalgebras generated by $GL(V)$ and $\mathbb{C}[S_n]$ inside the algebra of linear transformations of $W$, respectively. Then $A$ and $B$ are the centralizers of each other in the endomorphism algebra.
\label{DoubleCentThm}
\end{thm}

\begin{thm}
The vector space $W$ admits a canonical decomposition 

\begin{equation}
W=\underset{r(\lambda)\leq d}{\underset{\lambda\vdash n}{\bigoplus}} W_{\lambda}=\underset{r(\lambda)\leq d}{\underset{\lambda\vdash n}{\bigoplus}} \mathbb{S}_\lambda\otimes V_{\lambda},
\label{SchurWeyl}
\end{equation}

where $\lambda$ runs over partitions with at most $d$ constituents, $\mathbb{S}_\lambda$ and $V_{\lambda}$ are irreducible representations of $S_n$ and $GL(V)$ indexed by $\lambda$, respectively.
\end{thm}

\begin{ex}
 Let $V$ be the $3$-dimensional space spanned by qutrits $|0\rangle,|1\rangle$ and $|2\rangle$ and $W=V^{\otimes 7}$. The decomposition in \eqref{SchurWeyl} gives a splitting of $W$ into a direct sum of eight subspaces indexed by partitions of $7$ in at most three summands  (see Figure \ref{SubDivPic} for an illustration).
\end{ex}

\begin{figure}[htbp!]
	\begin{center}	
		\begin{tikzpicture}[scale=0.5,thick]
			\draw[thick, fill=brown!50] (0,0) -- (14,0) -- (14,10) -- (0,10) -- cycle;
			
    \node (B) at (13,5) {$W_{(7)}$};
    \node (B) at (1.2,1) {$W_{(6,1)}$};
   \node (B) at (9.3,4) {$W_{(5,2)}$};
   \node (B) at (12.8,0.7) {$W_{(5,1,1)}$};
   \node (B) at (3,4.7) {$W_{(4,3)}$};
   \node (B) at (2,8) {$W_{(4,2,1)}$};
   \node (B) at (11,9) {$W_{(3,3,1)}$};
   \node (B) at (5.3,2) {$W_{(3,2,2)}$};

\tikzset{decoration={snake,amplitude=.4mm,segment length=2mm,
                       post length=0mm,pre length=0mm}}
  \draw[decorate] (0,6)  to [bend right=45] (5,10);
  \draw[decorate] (5,10)  to [bend right=30] (14,7);
  \draw[decorate] (5,10)  to [bend right=30] (11,0);
  \draw[decorate] (4,0)  to [bend left=20] (0, 6);
  \draw[decorate] (11,0)  to [bend left=30] (14,3);
  \draw[decorate] (12.3,6.7)  to [bend right=80] (14,3);
  \draw[decorate] (4,0)  to [bend left=70] (6,5);

		\end{tikzpicture}
	\end{center}
 \caption{Decomposition of $W=V^{\otimes 7}$ for $d=3$}
 \label{SubDivPic}
\end{figure}
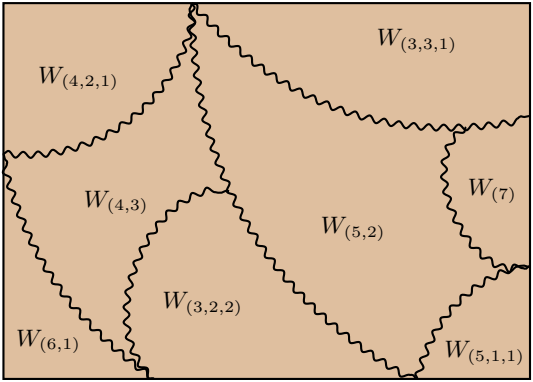

\section{Case study: $S_n\subseteq G$ or $SU_d(\mathbb{C})\subseteq G$}

In this section, we will develop a specific formula for a Hamiltonian $H_{M,\lambda}$ that fulfills the assumptions of Theorem \ref{MainObs} for each subspace $W_\lambda$ appearing in decomposition \eqref{SchurWeyl}. We proceed by deriving various corollaries of the fact that $H_{M,\lambda}$ with aforementioned properties exists.

\begin{defn}
Consider a partition $\lambda\vdash n$ with a maximum of $d$ rows. We define $T_\lambda\in SYT(\lambda)$ as the tableau of shape $\lambda=(\lambda_1,\hdots,\lambda_s)$, where the number $i+\sum\limits_{k=1}^{j-1}\lambda_k$ is assigned to the box at coordinates $(i,j)$. Let $T\subset SU_d(\mathbb{C})$  be a maximal torus and $\mathfrak{t}=\text{Lie}(T)$ be the corresponding Lie algebra. Set $v_{T_{\lambda}}\in \mathbb{S}_\lambda$ to be the vector associated with tableau $T_\lambda$, and $w_\lambda\in V_\lambda$ a minimal weight vector with respect to $\mathfrak{t}$.

Define the Hamiltonian 
\begin{equation}
\begin{aligned}
    &H_{M,\lambda} = \sum\limits_{i=1}^s \left(\mathcal{J}_{\lambda_1+\hdots+\lambda_i}-(\lambda_i-i)\cdot I\right)^{2}+\\
    &+ \varepsilon\sum\limits_{j=1}^n I\otimes\hdots\otimes\underset{j}{\tau}\otimes\hdots\otimes I,
\end{aligned}
\end{equation}
where $\tau=\text{diag}(-d/2,-d/2+1,\hdots,d/2)\in \mathfrak{t}$ 
with $0<\varepsilon<\frac{2}{nd}$.
\end{defn}

\begin{rmk}
The Hamiltonians 
\[
H_{M,\lambda}-\varepsilon\sum\limits_{j=1}^n I\otimes\hdots\otimes\underset{j}{\tau}\otimes\hdots\otimes I
\]
(up to addition of a constant multiple of identity operator) belong to the family of Hamiltonians introduced in Section III of \cite{Z1LLSK2}.
\end{rmk}

For the reader's convenience, we present a concrete example illustrating the derivation of a formula for the mixer Hamiltonian associated with a given Young diagram.

\begin{ex}
Consider a Young diagram with shape $\lambda=(3,3,2,1)$. The corresponding tableau, denoted as $T_\lambda$, is illustrated in Figure \ref{SYT}. In this case, the mixer Hamiltonian $H_{M,\lambda}$ is given by 
\begin{equation}
\begin{aligned}
&H_{M,\lambda}=  (\mathcal{J}_{3}-(3-1))^{2}+(\mathcal{J}_{6}-(3-2))^{2}+\\
&+(\mathcal{J}_{8}-(2-3))^{2}+(\mathcal{J}_{9}-(1-4))^{2}+\varepsilon\tau= \\ 
&=(\mathcal{J}_{3}-2)^{2}+(\mathcal{J}_{6}-1)^{2}+(\mathcal{J}_{8}+1)^{2}+(\mathcal{J}_{9}+3)^{2}+\varepsilon\tau,
\end{aligned}
\end{equation}
where $\varepsilon<\frac{1}{9d}$.
\end{ex}

\begin{figure}[htbp!]
    \centering
    $$\young(123,456,78,9)$$
    \caption{Standard Young tableaux $T_\lambda$ for diagram of shape $\lambda=(3,3,2,1)$}
    \label{SYT}
\end{figure}
 
\begin{thm} 
Every Hamiltonian  $H_{M,\lambda}$  meets the conditions specified in  Theorem \ref{MainObs}. The lowest-energy eigenspace of $H_{M,\lambda}$ is  spanned by the vector $\xi_\lambda:=v_{T_{\lambda}}\otimes w_\lambda$.
\label{MainThm}
\end{thm}

\begin{cor}
Suppose that the problem Hamiltonian $H_P$ commutes with the action of $SU_d(\mathbb{C})$ or $S_n$. Then for each diagram $Y_\lambda$ with no more than $d$ rows the state 
\[\nu_\lambda:=\lim\limits_{p\rightarrow \infty}(e^{-i\beta H_{M,\lambda}}e^{-i\beta H_P})^p (\xi_\lambda)
\]
emerges as a ground state of the problem Hamiltonian $H_P$.
\end{cor}

\begin{proof}
A direct consequence of part $(2)$ of Theorem \ref{MainThm} and the adiabatic theorem. 
\end{proof}

\begin{cor}
Suppose that the problem Hamiltonian $H_P$ commutes with the action of $SU_d(\mathbb{C})$ or $S_n$. Let $\theta$ be the number of classical states $s \in\mathbb{B}^n$, on which the minimal value of objective function $F$ is attained. Then $\theta$ is bounded below by the number of partitions $\lambda=(\lambda_!,\hdots,\lambda_s) \vdash n$ with $s\leq d$. 
\label{MinSolCor}
\end{cor}

\begin{cor}
Suppose that the problem Hamiltonian $H_P$ commutes with the action of $SU_d(\mathbb{C})$ or $S_n$. Then the original QAOA can be reduced to $QAOA_\lambda$ with ambient vector space $\mathbb{S}_{\lambda}\otimes W_\lambda$ of  dimension given by a polynomial function in $n$.
\end{cor}

\begin{ex}
One choice of such a diagram $\lambda$ is a hook diagram corresponding to partition $\lambda=(n-k,1,1,\hdots,1)$ with $0\leq k<d$ (see Figure \ref{HookRep}). 

The dimension of representation $\mathbb{S}_\lambda$ is 
\[
\frac{n!}{1\cdot 2\cdot\hdots\cdot (n-k-1)\cdot n\cdot k\cdot\hdots\cdot 1}=\frac{(n-1)!}{(n-k-1)!k!} ~~~\text{ (see \eqref{HLformula})},
\] 
while the dimension of $V_\lambda$  is given by a polynomial $\mathcal{P}_\lambda(n)$ of degree $d-1$ (follows from \eqref{DimFormula}). Hence, the dimension of $\mathbb{S}_{\lambda}\otimes V_\lambda$ is given by a polynomial of degree $d+k-1$ in $n$. 
\label{PolDim}
\end{ex}
\begin{figure}[htbp!]
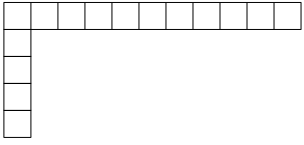

    \centering
    $\young(~~~~~~~~~~~,~,~,~,~)$
    \caption{Diagram of shape $\lambda=(n-k,1,1,\hdots,1)$}
    \label{HookRep}
\end{figure}

\section{Conclusion and potential applications}
In conclusion, our work introduces a novel approach leveraging representation theory principles for groups of symmetries in the problem Hamiltonian within the context of Quantum Approximate Optimization Algorithm. This approach leads to significant reductions in the Hilbert space, as demonstrated in Theorem \ref{MainObs}, potentially offering substantial practical implications for development of novel and efficient QAOA algorithms.

As an example, when the group of symmetries involves the symmetric or unitary group, we have successfully constructed multiple reductions of the Hilbert space for the original QAOA. Each of these reductions, denoted as $QAOA_\lambda$, preserves the problem Hamiltonian $H_P$, but varies in terms of the mixer Hamiltonian $H_{M,\lambda}$ and initial state $\xi_\lambda$, confined within the subspace $W_{\lambda}$.

As a result, our findings have direct relevance to classical simulation techniques. Recall that $W$ represents the Hilbert space required for an exact simulation of QAOA, and likewise, each $W_\lambda$ is the vector space for simulating $QAOA_\lambda$. As outlined in Section $3.2$ of \cite{SHHS}, efficient classical simulation becomes feasible under the assumption that the matrix elements of $H_{M,\lambda}$ and $H_P$ (as operators on $W_\lambda$) can be computed with reasonable speed.

Additionally, our findings may offer valuable insights into addressing the prominent challenge known as the barren plateau problem in parameterized quantum circuits. This phenomenon, characterized by exponentially diminishing gradients of the cost function with respect to the parameters as the number of qubits increases, poses a significant obstacle to classical optimization algorithms during the training process. By considering reduced QAOAs, we present an additional avenue to cope with  this challenge, complementing existing strategies proposed in the literature (see \cite{GWOB,kulshrestha2022beinit, MBSN} and references therein).

In summary, our contributions not only enhance understanding of the theoretical underpinnings of QAOA but also open up new avenues for efficient classical circuit simulations and offer potential solutions to the barren plateau problem, thereby advancing the practical applicability of quantum optimization strategies.

\section*{Acknowledgements}
This work was supported  in part with funding from the Defense Advanced Research Projects Agency (DARPA). The views, opinions and/or  findings expressed are those of the authors and should not be interpreted as representing the official views or policies  of the Department of Defense or the U.S. Government.
Y.A. acknowledges support from the U.S. Department of Energy, Office of Science, under contract DE-AC02-06CH11357 at Argonne National Laboratory. 

\section*{Appendix: The Proofs}
We now provide rigorous verifications of the results announced in the previous sections.

Theorem \ref{HamThm} asserted that the group of symmetries of  
Hamiltonian  
\[
H_P=\sum\limits_{k}\widetilde{\beta}_kZ_k+\sum\limits_{i,j}\widetilde{\alpha}_{i,j}Z_iZ_j
\]
consists of unitary matrices, which are block-diagonal in the standard basis: \[   \left(\begin{array}{ c|c|c|c }
    M_1 & 0 & \hdots & 0 \\
    \hline
    0 & M_2 & \ddots & 0 \\
    \hline
    \vdots & \ddots & \ddots  & \vdots\\
    \hline
    0 & 0 & \hdots & M_t
  \end{array}\right)\subseteq SU_N(\mathbb{C}).
\]

\begin{proof}
The action of $H_P$ on a standard basic vector $v_b$ is given by $H_P(v_b)=\lambda_bv_b$, where $\lambda_b=\sum\limits_{b_i\equiv b_j} \widetilde{\alpha}_{ij}-\sum\limits_{b_i\not\equiv b_j} \widetilde{\alpha}_{ij} + \sum\limits_{b_i\equiv 0} \widetilde{\beta}_{i}-\sum\limits_{b_j\equiv 1} \widetilde{\beta}_{j}$. Hence, in the standard basis  the matrix of problem Hamiltonian is block-diagonal, 
\[
H_P=
 \left(\begin{array}{ c|c|c|c }
    \lambda_1\cdot I & 0 & \hdots & 0 \\
    \hline
    0 & \lambda_2\cdot I & \ddots & 0 \\
    \hline
    \vdots & \ddots & \ddots  & \vdots\\
    \hline
    0 & 0 & \hdots & \lambda_t\cdot I
  \end{array}\right)\]
  with $\lambda_s\in\{\lambda_b~|~b \in\mathbb{B}^n \}$. Here $t$ is the number of different values attained by $\lambda_b$'s. The matrices that commute with $H_P$, written in the standard basis, have  block-diagonal form indicated in the theorem.
\end{proof}

Next we present the proof of Theorem \ref{MainThm} and its Corollary \ref{MinSolCor} giving a lower bound on the number of classical states $s \in\mathbb{B}^n$, on which the minimal value of objective function $F$ is attained.

Recall that the theorem consists of three assertions on Hamiltonian 
\[\begin{aligned}
    &H_{M,\lambda}= \sum\limits_{i=1}^s \left(\mathcal{J}_{\lambda_1+\hdots+\lambda_i}-(\lambda_i-i)\cdot I\right)^{2} +\\
    &  \varepsilon\sum\limits_{j=1}^n I\otimes\hdots\otimes\underset{j}{\tau}\otimes\hdots\otimes I:
\end{aligned}\]

\begin{enumerate}
\item the action of $H_{M,\lambda}$ preserves each direct summand of \eqref{SchurWeyl};

\item the lowest-energy eigenspace of the Hamiltonian $H_{M,\lambda}$ is one-dimensional and spanned by the vector $\xi_\lambda=v_{T_{\lambda}}\otimes w_\lambda$.
\item $H_{M,\lambda}$ satisfies the conditions of Theorem \ref{PF}.
\end{enumerate}

The first statement follows from the fact that $H_{M,\lambda}$ is a sum of an element in the group algebra $\mathbb{C}[S_n]$ and an element in $U(\mathfrak{su}_d)$, the universal enveloping algebra of $SU_d(\mathbb{C})$. 

To confirm the second assertion, we proceed in two stages. Initially, we consider the Hamiltonian $\widetilde{H}_{M,\lambda}=H_{M,\lambda}-\varepsilon\sum\limits_{j=1}^n I\otimes\hdots\otimes\underset{j}{\tau}\otimes\hdots\otimes I$, which is an element residing in  $\mathbb{C}[S_n]$. It follows that $\widetilde{H}_{M,\lambda}$ acts trivially on each $V_\lambda$ in decomposition \eqref{SchurWeyl}. The subspace of minimal energy for $\widetilde{H}{M,\lambda}$ is identified as $v{T_{\lambda}}\otimes V_\lambda$. This is demonstrated by the condition $\widetilde{H}{M,\lambda}(v_T)=0$, which fixes the positions of boxes labeled with numbers $\lambda_1,\lambda_1+\lambda_2,\hdots,\lambda_1+\hdots+\lambda_s$. In turn, there is a unique Young tableau $T=T_\lambda$ satisfying the proposed arrangement of aforementioned  boxes. Hence, $v_{T_{\lambda}}\otimes V_\lambda$ is the lowest-energy eigenspace of $\widetilde{H}_{M,\lambda}$. Next we notice that the absolute values of eigenvalues of the operator $\varepsilon\sum\limits_{j=1}^n I\otimes\hdots\otimes\underset{j}{\tau}\otimes\hdots\otimes I$ on $W$ are bounded above by $\varepsilon \cdot n\cdot\frac{d}{2}<\frac{2dn}{2dn}=1$. As the eigenvalues of $\widetilde{H}_{M,\lambda}$ are positive integers (see Remark \ref{YJMevals}), the lowest-energy eigenvector of $H_{M,\lambda}$ coincides with the lowest weight vector for $\tau$ in $v_{T_{\lambda}}\otimes V_\lambda$. The latter is exactly $v_{T_{\lambda}}\otimes w_\lambda$.    

The statement in the corollary can be verified as follows. For each $Y_\lambda$ with at most $d$ rows, the vector $\nu_\lambda=\lim\limits_{p\rightarrow \infty}(e^{-i\beta H_{M,\lambda}}e^{-i\beta H_P})^p (\xi_\lambda)$ is an eigenvector of $H_P$ with minimal eigenvalue. The associated eigenspace $V_{\min}\subseteq W$ is spanned by classical states that correspond to minima of the objective function $F$ on $\mathbb{D}^n$. Consequently, $\nu_\lambda \in V_{\min}$.  It remains to notice that the vectors $\nu_\lambda$ are linearly independent.



\begin{thebibliography}{99}


\bib{BFL}{article}{
   author={Bourreau, E.},
   author={Fleury, G.},
   author={Lacomme, P.},
   title={Mixer Hamiltonian with QAOA for Max $k$-coloring: numerical evaluations},
   journal = {arXiv:2207.11520},
   date = {2022}
}

\bib{egger2021warm}{article}{
  title={Warm-starting quantum optimization},
  author={Egger, Daniel J and Mare{\v{c}}ek, Jakub and Woerner, Stefan},
  journal={Quantum},
  volume={5},
  pages={479},
  year={2021},
  publisher={Verein zur F{\"o}rderung des Open Access Publizierens in den Quantenwissenschaften}
}


\bib{BKZS}{article}{
   author={Binkowski, L.},
   author={Ko{\ss}mann, G.},
   author={Ziegler, T.},
   author={Schwonnek, R.},
   title={Elementary Proof of QAOA Convergence},
journal = {arXiv:2302.04968},
   date = {2023}
}

\bib{CGLW}{article}{
  title={Symmetry Protected Topological Orders and the Classification of Bosonic Symmetry-Protected Phases},
  author={Chen, X.},
  author={Gu, Z.},
  author={Liu, Z.},
  author={Wen, X.},
  journal={Physical Review B},
  volume={82},
  number={15},
  year={2010},
  publisher={APS}
}

\bib{Et}{book} {
    author = {Etingof, P.}, 
    author = {Golberg, O.},
    author = {Hensel, S.},
    author = {Liu, T.},
    author = {Schwendner, A.},
    author = {Vaintrob, D.},
     TITLE = {Introduction to representation theory},
    SERIES = {Student Mathematical Library},
    VOLUME = {59},
      NOTE = {With historical interludes by Slava Gerovitch},
 PUBLISHER = {American Mathematical Society, Providence, RI},
      YEAR = {2011},
     PAGES = {viii+228},
      ISBN = {978-0-8218-5351-1},
}
\bib{FH}{book}{
   author={Fulton, W.},
   author={Harris, J.},
   title={Representation theory},
   series={Graduate Texts in Mathematics},
   volume={129},
   note={A first course;
   Readings in Mathematics},
   publisher={Springer-Verlag, New York},
   date={1991},
   pages={xvi+551},
   isbn={0-387-97527-6},
}

\bib{GBOE}{article}{
   author={Golden, J.},
   author={B{\"a}rtschi, A.},
   author={O’Malley, D.},
   author={Eidenbenz, S.},
   title={Numerical Evidence for Exponential Speed-up of
QAOA over Unstructured Search for Approximate
Constrained Optimization},
    journal = {arXiv:2202.00648},
   date = {2022}
}

\bib{GC}{article}{
   author={Gelfand, I. M.},
   author={Cetlin, M.},
   title={Finite-dimensional representations of the group of unimodular
 matrices.},
   language={Russian},
   journal={Doklady Akad. Nauk SSSR (N.S.)},
   date={1950},
   pages={825--828},
}

\bib{Got}{article}{
  title={Class of Quantum Error-Correcting Codes Saturating the Quantum Hamming Bound},
  author={Gottesman, D.},
  journal={Physical Review A},
  volume={54},
  number={3},
  pages={1862--1868},
  year={1996},
  publisher={APS}
}

\bib{GPSK}{article}{
   author={Govia, L. C. G.},
   author={Poole, C.},
   author={Saffman, M},
   author={Krovi, H. K},
   title={Freedom of the mixer rotation axis improves performance in the quantum approximate optimization algorithm.},
   journal={Phys. Rev. A},
   volume={104},
   date={2021},
   number={6},
}

\bib{GHPGA}{article}{
  title={QuYBE-An Algebraic Compiler for Quantum Circuit Compression},
  author={Gulania, S.},
  author={He, Z.}, 
  author={Peng, B.}, 
  author={Govind, N.},
  author={Alexeev, Y.},
  booktitle={2022 IEEE/ACM 7th Symposium on Edge Computing (SEC)},
  pages={406--410},
  year={2022},
  organization={IEEE}
}


\bib{Had}{article}{
   author={Hadfield, S.},
   title={On the representation of Boolean and real functions as
   Hamiltonians for quantum computing},
   journal={ACM Trans. Quantum Comput.},
   volume={2},
   date={2021},
   number={4},
   pages={Art. No. 18, 21},
   issn={2643-6809},
}


\bib{Has}{article}{
   author={Hastings, M.},
   title={Classical and quantum bounded depth approximation algorithms},
    journal = {arXiv:1905.07047},
   date = {2019}
}


\bib{HWORVB}{article}{
   author={Hadfield, S.},
   author={Wang, Z.},
   author={O’Gorman, B.},
   author={Rieffel, E.G.},
   author={Venturelli, D.},
   author={Biswas, R.},
   title={From the quantum approximate optimization algorithm to a quantum alternating operator ansatz},
   journal={Algorithms},
   date={2019},
   volume={12},
   number={34},
}

\bib{kulshrestha2022beinit}{inproceedings}{
  title={{BEINIT}: Avoiding barren plateaus in variational quantum algorithms},
  author={Kulshrestha, A.},
  author={Safro, I.},
  booktitle={2022 IEEE International Conference on Quantum Computing and Engineering (QCE)},
  pages={197--203},
  year={2022},
  organization={IEEE}
}

\bib{KSCAB}{article}{
   author={Khairy, S.},
   author={Shaydulin, R.},
   author={Cincio, L.},
   author={Alexeev, Y.},
   author={Balaprakash, P.},
   title={Learning to optimize variational quantum circuits to solve combinatorial problems},
   journal={Proceedings of the AAAI Conference on Artificial Intelligence},
   organization={IEEE},
   volume={34},
   date={2020},
   number={3},
   pages={2367--2375}
}

\bib{K}{article}{
   author={Krovi, H.},
   title={An efficient high dimensional quantum Schur transform},
   journal={Quantum},
   volume={3},
   pages={122},
   date={2019},
}

\bib{QAOA}{article}{
   author={Farhi, E.},
   author={Goldstone, J.},
   author={Gutmann, S.},
   title={A quantum approximate optimization algorithm},
   journal = {arXiv:1411.4028},
   date = {2014}
}

\bib{GWOB}{article}{
  title={An initialization strategy for addressing barren plateaus in parametrized quantum circuits},
  author={Grant, E.},
  author={Wossnig, L.},
  author={Ostaszewski, M.},
  author={Benedetti, M.},
  journal={Quantum},
  volume={3},
  pages={214},
  year={2019}
}


\bib{Ll}{article}{
   author={Lloyd, S.},
   title={Quantum approximate optimization is computationally universal},
   journal = {arXiv:1812.11075},
   date = {2018}
}

\bib{MBSN}{article}{
  title={Barren plateaus in quantum neural network training landscapes},
  author={McClean, J. R},
  author={Boixo, S.},
  author={Smelyanskiy, V. N.}, 
  author={Neven, H.},
  journal={Nature communications},
  volume={9},
  number={1},
  pages={4812},
  year={2018}
}

\bib{MBZ}{article}{
   author={Morales, M. E. S.},
   author={Biamonte, J. D.},
   author={Zimborás, Z.},
   title={On the universality of the quantum approximate optimization algorithm},
   journal={Quantum Information Processing},
   volume={19},
   date={2020},
   number={291},
}

\bib{OV}{article}{
   author={Okounkov, A.},
   author={Vershik, A.},
   title={A new approach to representation theory of symmetric groups},
   journal={Selecta Math. (N.S.)},
   volume={2},
   date={1996},
   number={4},
   pages={581--605},
   issn={1022-1824},
}
\bib{PGAG}{article}{
  title={Quantum time dynamics employing the Yang-Baxter equation for circuit compression},
  author={Peng, B.}, 
  author={Gulania, S.},
  author={Alexeev, Y.},
  author={Govind, N.},
  journal={Physical Review A},
  volume={106},
  number={1},
  pages={012412},
  year={2022},
  publisher={APS}
}




\bib{SA}{article}{
   author={Shaydulin, R.},
   author={Alexeev, Y.},
   title={Evaluating quantum approximate optimization algorithm: A case study},
   journal={Tenth International Green and Sustainable Computing Conference (IGSC)},
   organization={IEEE}
   date={2019},
}

\bib{SBCDLP}{article}{
  title={The Security of Practical Quantum Key Distribution},
  author={Scarani, V.},
  author={Bechmann-Pasquinucci, H.},
  author={Cerf, N.},
  author={Du{\v{s}}ek, M.},
  author={L{\"u}tkenhaus, N.},
  author={Peev, M.},
  journal={Rev. Mod. Phys.},
  volume={81},
  number={3},
  pages={1301--1350},
  year={2009},
  publisher={APS}
}



\bib{SHHS}{article}{
   author={Shaydulin, R.},
   author={Hadfield, S.},
   author={Hogg, T.},
   author={Safro, I.},
   title={Classical symmetries and the quantum approximate optimization
   algorithm},
   journal={Quantum Inf. Process.},
   volume={20},
   date={2021},
   number={11},
   pages={Paper No. 359, 28},
}


\bib{SG}{article}{
   author={Shaydulin, R.},
   author={Galda, A.},
   title={Error mitigation for deep quantum optimization circuits by leveraging problem symmetries},
   journal = {arXiv:2106.04410},
   date = {2021}
}

\bib{shaydulin2019multistart}{inproceedings}{
  title={Multistart methods for quantum approximate optimization},
  author={Shaydulin, R.},
  author={Safro, I.},
  author={Larson, J.},
  booktitle={2019 IEEE high performance extreme computing conference (HPEC)},
  pages={1--8},
  year={2019},
  organization={IEEE}
}

\bib{shaydulin2019hybrid}{article}{
  title={A hybrid approach for solving optimization problems on small quantum computers},
  author={Shaydulin, R.},
  author={Ushijima-Mwesigwa, H.},
  author={Safro, I.},
  author={Mniszewski, S.},
  author={Alexeev, Y.},
  journal={Computer},
  volume={52},
  number={6},
  pages={18--26},
  year={2019},
  publisher={IEEE}
}

\bib{herman2023quantum}{article}{
  title={Quantum computing for finance},
  author={Herman, Dylan},
  author={Googin, Cody},
  author={Liu, Xiaoyuan},
  author={Sun, Yue},
  author={Galda, Alexey},
  author={Safro, Ilya},
  author={Pistoia, Marco},
  author={Alexeev, Yuri},
  journal={Nature Reviews Physics},
  pages={1--16},
  year={2023},
  publisher={Nature Publishing Group UK London}
}

\bib{boulebnane2023peptide}{article}{
  title={Peptide conformational sampling using the quantum approximate optimization algorithm},
  author={Boulebnane, Sami},
  author={Lucas, Xavier},
  author={Meyder, Agnes},
  author={Adaszewski, Stanislaw},
  author={Montanaro, Ashley},
  journal={npj Quantum Information},
  volume={9},
  number={1},
  pages={70},
  year={2023},
  publisher={Nature Publishing Group UK London}
}

\bib{ushijima2021multilevel}{article}{
  title={Multilevel combinatorial optimization across quantum architectures},
  author={Ushijima-Mwesigwa, H.},
  author={Shaydulin, R.},
  author={Negre, C.},
  author={Mniszewski, S.},
  author={Alexeev, Y.},
  author={Safro, I.},
  journal={ACM Transactions on Quantum Computing},
  volume={2},
  number={1},
  pages={1--29},
  year={2021},
  publisher={ACM New York, NY, USA}
}

\bib{simtransfer}{article}{
author = {Galda, A.},
author = {Gupta, E.},
author = {Falla, J.},
author = {Liu, X.},
author = {Lykov, D.},
author = {Alexeev, Y.},
author = {Safro, I.},   	 
title = {Similarity-based parameter transferability in the quantum approximate optimization algorithm},     
journal = {Frontiers in Quantum Science and Technology},      
volume = {2},           
year = {2023},      
}

\bib{SW}{article}{
   author={Shaydulin, R.},
   author={Wild, S.},
   title={Exploiting symmetry reduces the cost of training QAOA},
   journal={Transactions on Quantum Engineering},
   organization={IEEE},
   date={2021},
   number={2},
   pages={1-9},
}


\bib{VO}{article}{
   author={Vershik, A. M.},
   author={Okunkov, A. Yu.},
   title={A new approach to representation theory of symmetric groups. II},
   language={Russian, with English and Russian summaries},
   journal={Zap. Nauchn. Sem. S.-Peterburg. Otdel. Mat. Inst. Steklov.
   (POMI)},
   volume={307},
   date={2004},
   pages={57--98, 281},
   issn={0373-2703},
   translation={
      journal={J. Math. Sci. (N.Y.)},
      volume={131},
      date={2005},
      number={2},
      pages={5471--5494},
      issn={1072-3374},
   },
}

\bib{VBMSBJNM}{article}{
   author={Verdon, G.},
   author={Broughton, M.},
   author={McClean, J. R.},
   author={Sung, K. J.},
   author={Babbush, R.},
   author={Jiang, Z.},
   author={Neven, H.},
   author={Mohseni, M.},
   title={Learning to learn with quantum neural networks via
classical neural networks},
    journal = {arXiv:1907.05415},
   date = {2019}
}


 

\bib{WHJR}{article}{
   author={Wang, Z.},
   author={Hadfield, S.},
   author={Jiang, Z.},
   author={Rieffel, E.},
   title={Quantum approximate optimization algorithm for MaxCut: A fermionic vier},
   journal={Physical Review A},
   volume={97},
   date={2018},
   number={2},
}

\bib{ZLCMBE}{article}{
   author={Zhu, L.},
   author={Lun Tang, H.},
   author={Calderon-Vargas, F. A.},
   author={Mayhall, N.},
   author={Barnes, E.},
   author={Economou, S.},
   title={Adaptive quantum approximate optimization algorithm for solving combinatorial problems on a quantum computer},
   journal={Phys. Rev. Research.},
   volume={4},
   date={2022},
   number={3},
}

\bib{Z1LLSK1}{article}{
   author={Zheng, H.},
   author={Li, Z.},
   author={Liu, J.},
   author={Strelchuk, S.},
   author={Kondor, R.},
   title={On the super-exponential quantum speedup of equivariant quantum machine learning algorithms with $SU(d)$ symmetry},
    journal = {arXiv:2207.07250},
   date = {2022}
}

\bib{Z1LLSK2}{article}{
   author={Zheng, H.},
   author={Li, Z.},
   author={Liu, J.},
   author={Strelchuk, S.},
   author={Kondor, R.},
   title={Speeding Up Learning Quantum States Through Group Equivariant
Convolutional Quantum Ans\"{a}tze},
   journal={PRX QUANTUM},
   date={2023},
   number={4},
   pages={1--18},
}

\bib{ZWCPL}{article}{
   author={Zhou, L.},
   author={Wang, S.},
   author={Choi, S.},
   author={Pichler, H.},
   author={Lukin, M.},
   title={Quantum Approximate Optimization Algorithm: Performance, Mechanism, and Implementation on Near-Term Devices},
   journal={Phys. Rev. X},
   volume={10},
   date={2020},
}
\end{thebibliography}
\end{document}